\documentclass[sigconf]{aamas}

\usepackage{balance}
\usepackage{graphicx}
\usepackage{amsmath}
\usepackage{algorithm}
\usepackage{algorithmic}
\usepackage{booktabs}
\usepackage{tikz}
\usepackage{pgfplots}
\pgfplotsset{compat=1.18}
\usetikzlibrary{arrows,positioning,shapes}
\usepackage{multirow}
\usepackage{xcolor}
\usepackage{hyperref}

\setcopyright{ifaamas}
\acmConference[AAMAS '26]{Proc.\@ of the 25th International Conference
on Autonomous Agents and Multiagent Systems (AAMAS 2026)}{May 25 -- 29, 2026}
{Paphos, Cyprus}{C.~Amato, L.~Dennis, V.~Mascardi, J.~Thangarajah (eds.)}
\copyrightyear{2026}
\acmYear{2026}
\acmDOI{}
\acmPrice{}
\acmISBN{}

\acmSubmissionID{<<submission id>>}

\title[G²CP: Graph-Grounded Communication Protocol]{G²CP: A Graph-Grounded Communication Protocol for Verifiable and Efficient Multi-Agent Reasoning}

\author{Karim Ben Khaled}
\affiliation{
  \institution{University of Lorraine, LORIA}
  \city{Nancy}
  \country{France}}
\email{karim.ben-khaled@univ-lorraine.fr}

\author{Davy Monticolo}
\affiliation{
  \institution{University of Lorraine, LORIA}
  \city{Nancy}
  \country{France}}
\email{davy.monticolo@univ-lorraine.fr}

\begin{document}

\pagestyle{fancy}
\fancyhead{}


\begin{abstract}
Multi-agent systems powered by Large Language Models face a critical challenge: agents communicate through natural language, leading to semantic drift, hallucination propagation, and inefficient token consumption. We propose G²CP (Graph-Grounded Communication Protocol), a structured agent communication language where messages are graph operations rather than free text. Agents exchange explicit traversal commands, subgraph fragments, and update operations over a shared knowledge graph, enabling verifiable reasoning traces and eliminating ambiguity. We validate G²CP within an industrial knowledge management system where specialized agents (Diagnostic, Procedural, Synthesis, and Ingestion) coordinate to answer complex queries. Experimental results on 500 industrial scenarios and 21 real-world maintenance cases show that G²CP reduces inter-agent communication tokens by 73\%, improves task completion accuracy by 34\% over free-text baselines, eliminates cascading hallucinations, and produces fully auditable reasoning chains. G²CP represents a fundamental shift from linguistic to structural communication in multi-agent systems, with implications for any domain requiring precise agent coordination. Code, data, and evaluation scripts are publicly available.\footnote{\url{https://github.com/karim0bkh/G2CP_AAMAS}}
\end{abstract}

\maketitle



\section{Introduction}

LLM-based multi-agent systems have adopted natural language as their default communication medium~\cite{wang2024survey}. Frameworks such as AutoGen~\cite{wu2023autogen}, LangGraph~\cite{langgraph2024}, and CrewAI~\cite{crewai2024} enable developers to decompose complex tasks across specialized agents that collaborate through dialogue. While this linguistic flexibility lowers the barrier to multi-agent design, it introduces critical vulnerabilities: \textbf{semantic drift}, where meaning degrades across successive agent reinterpretations; \textbf{hallucination propagation}, where fabricated claims in one agent's output corrupt downstream reasoning; and \textbf{computational waste}, where verbose natural language exchanges consume tokens far in excess of their information content.

Consider a diagnostic scenario in industrial maintenance: Agent~A identifies a ``bearing failure'' and asks Agent~B to ``find the repair procedure.'' Agent~B interprets ``bearing failure'' as a generic bearing issue, retrieves a procedure for a different bearing type, and Agent~C estimates costs based on the wrong parts list. This cascading error arises because natural language lacks the precision needed for reliable inter-agent communication. The root causes are:

\begin{itemize}
\item \textbf{Referential ambiguity}: ``The main pump'' could refer to multiple entities.
\item \textbf{Temporal ambiguity}: ``Recent failures'' has no precise temporal scope.
\item \textbf{Relational ambiguity}: ``Connected to X'' does not specify relationship type.
\item \textbf{Contextual drift}: Each agent re-embeds and reinterprets previous messages.
\end{itemize}

We propose a radical alternative: \textbf{agents should communicate through structured graph operations rather than natural language}. Our contribution, the \textbf{Graph-Grounded Communication Protocol (G²CP)}, replaces textual messages with explicit commands over a shared knowledge graph. Instead of Agent~A saying ``find repair procedures for bearing B-4521,'' it sends:

\texttt{TRAVERSE FROM \{Part:B-4521\} VIA \{addressed\_by\} FOR 1 RETURN SUBGRAPH}

This message is unambiguous, executable, and verifiable. Agent~B performs the exact graph operation, returns a subgraph fragment, and the entire reasoning chain becomes auditable.

\paragraph{Intuition.}
The core idea is simple: if two agents share a knowledge graph, they can point to \emph{exactly} which nodes and edges they mean instead of describing them in words. A G²CP message is like handing a colleague a database query rather than an email---there is no room for misinterpretation. The protocol wraps these queries in classical performatives (REQUEST, INFORM, etc.) so that agents retain the social coordination mechanisms pioneered by FIPA-ACL~\cite{fipa2002acl}, but ground every content expression in the graph rather than in predicate logic or free text.

\subsection{Contributions}

This paper makes four primary contributions:

\begin{enumerate}
\item \textbf{The G²CP Protocol}: A formal agent communication language grounded in graph operations, including syntax, semantics, commitment semantics, and an explicit comparison with FIPA-ACL (Section~\ref{sec:protocol}).

\item \textbf{Multi-Agent Architecture}: A system design where specialized agents coordinate exclusively through G²CP, with LLM-based dynamic operation selection (Section~\ref{sec:architecture}).

\item \textbf{Experimental Validation}: Comprehensive evaluation on 500 synthetic and 21 real-world industrial scenarios, with full baseline implementation details and reproducibility artifacts (Section~\ref{sec:experiments}).

\item \textbf{Formal Analysis}: Proofs of protocol properties including determinism, auditability, and completeness guarantees, with a security and trust model (Section~\ref{sec:theory}).
\end{enumerate}

Our work bridges agent communication languages from classical AI~\cite{fipa2002acl,finin1994kqml,singh1998semantics,fornara2004communicative} with modern graph-augmented LLM systems, creating a foundation for reliable multi-agent reasoning.


\section{Related Work}
\label{sec:related}

\subsection{Agent Communication Languages}

Agent communication has a rich history predating LLMs. KQML~\cite{finin1994kqml} established knowledge-level performatives for inter-agent messaging. FIPA-ACL~\cite{fipa2002acl} refined these into a standard with formal semantics grounded in mental attitudes (beliefs, desires, intentions). Both assumed symbolic agents with shared ontologies and used content languages such as SL and KIF.

A significant line of work moved beyond mentalistic semantics toward \textbf{commitment-based approaches}. Singh~\cite{singh1998semantics} proposed social commitments as the basis for agent communication, arguing that meaning should be grounded in publicly observable social facts rather than private mental states. Yolum and Singh~\cite{yolum2002flexible} developed commitment protocols that enable flexible, exception-handling agent interactions. Fornara and Colombetti~\cite{fornara2004communicative} formalized communicative acts in terms of commitments and their lifecycles. Winikoff et al.~\cite{winikoff2005declarative} introduced declarative approaches to agent communication that further decoupled protocol specification from implementation.

Modern LLM-based agents typically abandon structured protocols in favor of natural language~\cite{wang2024survey}. While this enables flexibility, it sacrifices precision and verifiability. Our work resurrects structured communication but grounds it in graph operations rather than predicate logic, making it compatible with neural retrieval systems while preserving the social commitment guarantees of classical ACLs.

Table~\ref{tab:acl_comparison} provides a systematic comparison across content language, semantic grounding, neural compatibility, and verifiability.

\begin{table}[t]
\centering
\caption{Comparison of agent communication approaches}
\label{tab:acl_comparison}
\small
\begin{tabular}{@{}lcccc@{}}
\toprule
\textbf{Feature} & \textbf{KQML} & \textbf{FIPA} & \textbf{Commit.} & \textbf{G²CP} \\
\midrule
Content lang. & KIF & SL/KIF & Logic & GraphOp \\
Semantics & Mentalistic & BDI & Social & Graph \\
Neural compat. & \texttimes & \texttimes & \texttimes & \checkmark \\
Embeddings & \texttimes & \texttimes & \texttimes & \checkmark \\
Verifiable & Partial & Partial & \checkmark & \checkmark \\
Deterministic & \texttimes & \texttimes & Partial & \checkmark \\
Token-efficient & N/A & N/A & N/A & \checkmark \\
\bottomrule
\end{tabular}
\end{table}

\subsubsection{Semantic Heterogeneity: FIPA Ontologies vs.\ G²CP}
\label{sec:heterogeneity}

A core challenge for classical ACLs is \textbf{semantic heterogeneity}: when multiple agents maintain different ontologies, aligning symbolic definitions at runtime requires ontology matching~\cite{euzenat2007ontology}. FIPA addressed this through ontology services and content language negotiation, but practical deployments often struggled with interoperability.

G²CP sidesteps this problem by design. All agents operate on a \textbf{single shared graph instance} where entities are resolved once, during graph construction, through an entity resolution pipeline (string normalization, embedding-based fuzzy matching, domain-expert validation). Agent specializations represent different \emph{views} of the same graph---selecting different edge types for traversal---rather than different ontologies. This means inter-agent messages reference identical node identifiers, eliminating the need for runtime alignment.

\subsection{Multi-Agent LLM Systems}

Recent frameworks demonstrate impressive capabilities through natural language coordination. AutoGen~\cite{wu2023autogen} enables conversational workflows between specialized agents. LangGraph~\cite{langgraph2024} structures agent interaction as state machines with natural language transitions. CrewAI~\cite{crewai2024} assigns roles and goals to agents that collaborate through dialogue. MetaGPT~\cite{hong2024metagpt} uses structured outputs like documents and diagrams for coordination, and ChatDev~\cite{qian2024chatdev} simulates a software company with role-specialized agents.

However, these systems suffer from semantic drift. Liang et al.~\cite{liang2024debate} found that hallucination rates compound across agent chains, with errors in one agent's output corrupting downstream reasoning. Guo et al.~\cite{guo2024multiagent} showed that 43\% of multi-agent failures stem from miscommunication rather than individual agent errors.

G²CP addresses this by eliminating linguistic interpretation between agents entirely. Agents still use LLMs for understanding user queries and generating final responses, but inter-agent messages are deterministic graph operations.

\subsection{Knowledge Graphs and Agent Systems}

Knowledge graphs provide structured knowledge representation~\cite{hogan2021knowledge}. Recent work integrates graphs with LLMs for retrieval-augmented generation~\cite{edge2024graphrag,pan2024unifying,luo2024rog}. However, these systems treat the graph as a passive data structure rather than a communication medium.

Our work is conceptually closest to blackboard architectures~\cite{nii1986blackboard} where agents coordinate through a shared data structure, and to the Contract Net Protocol~\cite{smith1980contract} where structured task announcements replace free-form negotiation. G²CP formalizes this: the knowledge graph serves as both the information repository and the communication substrate. Agents post graph operations rather than symbolic assertions, enabling both coordination and knowledge retrieval through a unified mechanism.


\section{The G²CP Protocol}
\label{sec:protocol}

\subsection{Formal Definition}

Let $G = (V, E, \Lambda, \Psi)$ be a heterogeneous directed knowledge graph where $V = \{v_1, \ldots, v_n\}$ is the set of vertices (entities), $E \subseteq V \times V$ is the set of directed edges (relationships), $\Lambda = \{\lambda_1, \ldots, \lambda_k\}$ is the set of node types, and $\Psi = \{\psi_1, \ldots, \psi_m\}$ is the set of edge types.

Each node $v_i$ has type $\lambda(v_i) \in \Lambda$, a dense vector embedding $\mathbf{x}_i \in \mathbb{R}^d$ (computed by a sentence encoder over the node's textual description, enabling semantic similarity search), and attributes $A_i$ (a key-value dictionary storing structured properties such as \texttt{\{name: "Bearing B-4521", serial: "SN-9042", install\_date: "2021-03-15"\}}). Each edge $e_{ij} = (v_i, v_j)$ has type $\psi(e_{ij}) \in \Psi$, weight $w_{ij} \in [0,1]$, and timestamp $t_{ij}$.

A \textbf{G²CP message} is a tuple:
\begin{equation}
m = \langle \mathit{sender}, \mathit{receiver}, \mathit{perf}, \mathit{op}, \mathit{ctx} \rangle
\end{equation}

where $\mathit{sender}, \mathit{receiver} \in \mathcal{A}$ (agent identifiers), $\mathit{perf} \in \mathcal{P}$ (performative from Table~\ref{tab:performatives}), $\mathit{op}$ is a graph operation (Definition~\ref{def:graphop}), and $\mathit{ctx} = (\mathit{conv\_id}, G_{\mathit{focus}})$ where $G_{\mathit{focus}} \subseteq G$ is the current conversation subgraph.

\paragraph{Intuition.} A G²CP message is an instruction slip between two colleagues who share the same database. The performative says \emph{what kind of speech act} is being performed (asking, telling, proposing), and the operation says \emph{exactly which data} is involved. Because both agents see the same graph, there is zero room for misinterpretation of the content.

\begin{table}[t]
\centering
\caption{G²CP performatives and their semantics}
\label{tab:performatives}
\small
\begin{tabular}{@{}lp{7cm}@{}}
\toprule
\textbf{Performative} & \textbf{Semantic Commitment} \\
\midrule
REQUEST & Sender asks receiver to execute $op$ and return results \\
INFORM & Sender asserts that $op$ has been executed with specified results \\
QUERY & Sender asks whether $op$ would return non-empty results \\
PROPOSE & Sender suggests $op$ as a candidate operation for consideration \\
CONFIRM & Sender validates receiver's previous operation result \\
REJECT & Sender indicates receiver's operation violated constraints \\
UPDATE & Sender commands modification to $G$ via $op$ \\
\bottomrule
\end{tabular}
\end{table}

\begin{definition}[Graph Operation]
\label{def:graphop}
A graph operation $op$ has one of two forms:
\begin{enumerate}
\item \textbf{Traversal}: $op = \mathtt{TRAVERSE}(V_s, \Psi_f, h, \mathit{ret})$ where $V_s \subseteq V$ is the source node set, $\Psi_f \subseteq \Psi$ is the edge type filter, $h \in \mathbb{N} \cup \{\infty\}$ is the hop depth, and $\mathit{ret} \in \{\mathtt{SUBGRAPH}, \mathtt{PATHS}, \mathtt{LEAVES}\}$ specifies return type.

\item \textbf{Update}: $op = \mathtt{UPDATE}(\Delta G)$ where $\Delta G = (\Delta V^+, \Delta V^-, \Delta E^+, \Delta E^-)$ specifies nodes and edges to add/remove.
\end{enumerate}
\end{definition}

\subsection{Operational Semantics}

\paragraph{Intuition.} A traversal starts at a set of source nodes and ``walks'' outward along edges of specified types, collecting everything reachable within a given number of hops. The return format controls whether the agent gets back the full subgraph explored, just the end-points, or the specific paths taken.

Formally, the semantics of a traversal operation is:
\begin{equation}
\mathcal{T}: V_s \times \Psi_f \times \mathbb{N} \times \{\mathtt{SUBGRAPH}, \mathtt{PATHS}, \mathtt{LEAVES}\} \rightarrow 2^{V \times E}
\end{equation}

Computed recursively as:
\begin{equation}
\mathcal{T}(V_s, \Psi_f, h, \mathit{ret}) = \begin{cases}
\mathit{extract}(V_s, \emptyset) & \text{if } h = 0 \\
\mathcal{T}(V_s \cup N(V_s, \Psi_f), \Psi_f, h-1, \mathit{ret}) & \text{otherwise}
\end{cases}
\end{equation}

where $N(V_s, \Psi_f) = \{v_j \mid \exists v_i \in V_s, (v_i, v_j) \in E, \psi(v_i, v_j) \in \Psi_f\}$ computes the filtered neighborhood, and $\mathit{extract}$ formats the result based on $\mathit{ret}$.

An update operation transforms the graph:
\begin{equation}
G' = \mathtt{UPDATE}(G, \Delta G) = (V', E', \Lambda, \Psi)
\end{equation}
where $V' = (V \cup \Delta V^+) \setminus \Delta V^-$ and $E' = (E \cup \Delta E^+) \setminus \Delta E^-$.

\subsection{Protocol Properties}

We state core properties here; full proofs appear in Appendix~\ref{app:proofs}.

\begin{theorem}[Determinism]
\label{thm:determinism}
For a fixed graph state $G$ and operation $op$, the result of executing $op$ is deterministic and independent of agent implementation.
\end{theorem}

\emph{Proof sketch.} Traversal operations are recursive set expansions with deterministic neighborhood functions. Update operations are set-theoretic transformations. Neither involves stochastic processes. Full proof in Appendix~\ref{app:proofs}. $\square$

\begin{theorem}[Auditability]
\label{thm:auditability}
Any agent conclusion reachable through G²CP can be verified by replaying the message sequence over the graph.
\end{theorem}

\emph{Proof sketch.} Each message specifies explicit graph operations. Given initial state $G_0$ and message sequence $\langle m_1, \ldots, m_k \rangle$, we reconstruct graph states $G_0, G_1, \ldots, G_k$ and verify consistency. Full proof in Appendix~\ref{app:proofs}. $\square$

\begin{theorem}[Completeness]
\label{thm:completeness}
G²CP can express any query answerable through graph traversal and retrieval-augmented generation.
\end{theorem}

\emph{Proof sketch.} Any graph-based reasoning task decomposes into node selection, edge traversal, and subgraph retrieval, each mapping to traversal operations. Complex queries compose via message sequences. Full proof in Appendix~\ref{app:proofs}. $\square$

\paragraph{Practical implication.}
For safety-critical applications (industrial maintenance, healthcare, aviation), these properties mean that every system recommendation can be independently audited by replaying the exact graph operations that produced it. Unlike free-text multi-agent systems where intermediate reasoning is opaque, G²CP provides a complete, deterministic, and verifiable audit trail from user query to final answer.

\subsection{Message Syntax}

We define a concrete serialization format for G²CP messages:

\begin{verbatim}
<sender:agent_id> TO <receiver:agent_id>
PERFORMATIVE: <perf>
CONVERSATION: <conv_id>
OPERATION:
  TRAVERSE
    FROM: <node_selector>
    VIA: <edge_types>
    DEPTH: <hop_count>
    RETURN: <format>
  [CONSTRAINTS: <additional_filters>]
\end{verbatim}

Node selectors support explicit IDs (\texttt{\{node\_123, node\_456\}}), type-based selection (\texttt{\{type:Fault, type:Component\}}), property-based filtering (\texttt{\{Symptom WHERE severity>0.8\}}), and contextual references (\texttt{\{CURRENT\_FOCUS\}}).

\subsection{Relationship to FIPA-ACL}
\label{sec:fipa_comparison}

G²CP operates at the same \textbf{communication layer} as FIPA-ACL---it defines how agents exchange messages, not how tasks are orchestrated. The architectural parallel is:

\begin{center}
\small
\begin{tabular}{@{}lll@{}}
\toprule
\textbf{Layer} & \textbf{FIPA-ACL} & \textbf{G²CP} \\
\midrule
Performative & INFORM, REQUEST, \ldots & INFORM, REQUEST, \ldots \\
Content lang. & SL / KIF predicates & Graph operations \\
Ontology & Agent-specific & Shared graph schema \\
Transport & ACL envelope & Kafka + signatures \\
\bottomrule
\end{tabular}
\end{center}

Both frameworks use performatives to express illocutionary force. The critical difference lies in the content language: FIPA uses logical predicates (e.g., \texttt{(price (good g1) 50)}) while G²CP uses graph operations (e.g., \texttt{TRAVERSE FROM \{g1\} VIA \{priced\_at\} DEPTH 1}). This substitution yields two advantages for LLM-based agents: (1)~graph operations are directly executable without a theorem prover, and (2)~results (subgraphs) naturally integrate with neural retrieval pipelines through node embeddings.

The advance for the MAS field is that G²CP bridges symbolic ACL principles---structured performatives, verifiable content, social commitments---with modern neural architectures. Symbolic ACLs like FIPA require agents to share logical ontologies and employ theorem provers for content interpretation, making them incompatible with LLM-based systems. G²CP preserves the formal guarantees (auditability, determinism) while using graph operations that LLM agents can generate, execute, and reason over natively.

\subsection{Commitment Semantics}
\label{sec:commitments}

Following Singh~\cite{singh1998semantics} and Fornara and Colombetti~\cite{fornara2004communicative}, we formalize the social commitments created by each G²CP performative. We write $C(x, y, p)$ for ``agent $x$ is committed to agent $y$ that condition $p$ holds.''

\paragraph{REQUEST$(A, B, \mathtt{TRAVERSE}(\ldots))$}

\begin{itemize}
\item \textbf{Pre-condition}: $A$ believes $B$ can execute the traversal (i.e., $B$'s role covers the required edge types).
\item \textbf{Post-condition}: $B$ commits to return the traversal result or a structured ERROR.
\item \textbf{Social commitment}: $C(B, A, \text{execute\_and\_return}(\mathit{op}))$.
\item \textbf{Verification}: $A$ can replay $\mathit{op}$ on $G$ and check result consistency.
\end{itemize}

\paragraph{INFORM$(A, B, \mathit{subgraph})$}

\begin{itemize}
\item \textbf{Pre-condition}: $A$ has executed an operation producing $\mathit{subgraph}$.
\item \textbf{Post-condition}: $A$ asserts that $\mathit{subgraph} \subseteq G$.
\item \textbf{Social commitment}: $C(A, B, \text{grounded}(\mathit{subgraph}, G))$.
\item \textbf{Verification}: $B$ checks $\mathit{subgraph} \subseteq G$ by re-executing the source operation.
\end{itemize}

\paragraph{QUERY$(A, B, \mathtt{TRAVERSE}(\ldots))$}

\begin{itemize}
\item \textbf{Pre-condition}: $A$ seeks existence information.
\item \textbf{Post-condition}: $B$ commits to truthful boolean response.
\item \textbf{Social commitment}: $C(B, A, \text{truthful\_response}(\mathit{op}))$.
\item \textbf{Verification}: $A$ replays $\mathit{op}$ and checks non-emptiness.
\end{itemize}

\paragraph{PROPOSE$(A, B, \mathit{op})$}

\begin{itemize}
\item \textbf{Pre-condition}: $A$ believes $\mathit{op}$ is relevant.
\item \textbf{Post-condition}: $B$ commits to evaluate and respond with CONFIRM or REJECT.
\item \textbf{Social commitment}: $C(B, A, \text{evaluate\_and\_respond}(\mathit{op}))$.
\end{itemize}

\paragraph{CONFIRM$(A, B, \mathit{result})$}

\begin{itemize}
\item \textbf{Pre-condition}: $A$ has verified $B$'s previous result.
\item \textbf{Post-condition}: $A$ endorses $\mathit{result}$ as correct.
\item \textbf{Social commitment}: $C(A, B, \text{verified}(\mathit{result}))$.
\end{itemize}

\paragraph{REJECT$(A, B, \mathit{op})$}

\begin{itemize}
\item \textbf{Pre-condition}: $A$ detects constraint violation in $B$'s operation or result.
\item \textbf{Post-condition}: $B$ must not act on the rejected result.
\item \textbf{Social commitment}: $C(A, B, \text{violated}(\mathit{op}, \mathit{constraint}))$.
\end{itemize}

\paragraph{UPDATE$(A, B, \Delta G)$}

\begin{itemize}
\item \textbf{Pre-condition}: $A$ is authorized for graph modification.
\item \textbf{Post-condition}: $B$ commits to apply $\Delta G$ after validation.
\item \textbf{Social commitment}: $C(B, A, \text{apply\_if\_valid}(\Delta G))$.
\item \textbf{Verification}: $A$ queries updated graph to confirm changes.
\end{itemize}

All commitments are publicly observable through the audit log, satisfying Singh's~\cite{singh1998semantics} criterion that meaning should be grounded in social facts rather than private mental states.


\section{Multi-Agent Architecture}
\label{sec:architecture}

\subsection{Agent Roles and Specializations}

We instantiate G²CP within an industrial knowledge management system with four specialized agents:

\paragraph{Diagnostic Agent ($A_D$):} Performs root cause analysis by traversing symptom-to-fault relationships. \\Specializes in $\Psi_{\text{diag}} = \{\mathtt{causes}, \mathtt{indicates}, \mathtt{correlates\_with}\}$.

\paragraph{Procedural Agent ($A_P$):} Retrieves maintenance procedures by traversing fault-to-action relationships. \\Specializes in $\Psi_{\text{proc}} = \{\mathtt{addressed\_by}, \mathtt{requires}, \mathtt{precedes}\}$.

\paragraph{Synthesis Agent ($A_S$):} Discovers patterns through historical data traversal. \\Specializes in $\Psi_{\text{synth}} = \{\mathtt{occurred\_in}, \mathtt{replaced\_in}, \mathtt{failed\_after}\}$.

The Synthesis agent discovers new relationships through \textbf{temporal graph traversal over historical work order data}. Specifically, $A_S$ computes co-occurrence frequencies between fault events and contextual conditions. For each pair $(f, c)$ where $f$ is a fault node and $c$ is a condition node (sensor reading, environmental factor), $A_S$ traverses all historical work orders containing $f$ and checks for temporal co-occurrence of $c$ within a configurable time window (default: 48 hours). If the co-occurrence frequency exceeds a threshold (e.g., $f=$``Part X failure'' and $c=$``temp\_anomaly'' co-occur in 15 out of 20 cases), $A_S$ proposes a new edge with confidence $15/20 = 0.75$ via an UPDATE message. This mechanism enables continuous knowledge graph enrichment from operational data.

\paragraph{Ingestion Agent ($A_I$):} Updates the graph based on new information. Uses UPDATE operations to maintain $G$.

\subsection{Coordination Protocol}

Figure~\ref{fig:sequence} shows a typical interaction sequence for a complex diagnostic query.

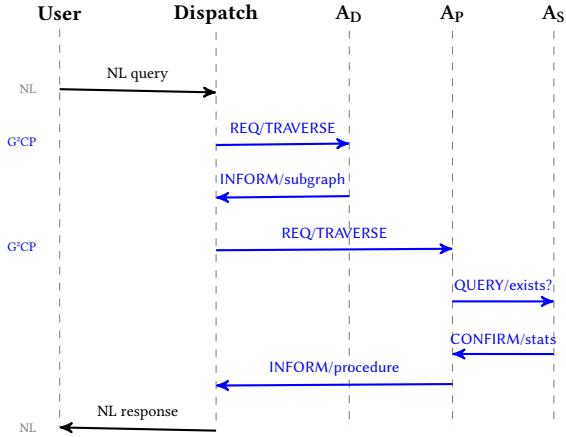
\begin{figure}[t]
\centering
\begin{tikzpicture}[node distance=1.5cm, >=stealth', font=\small]
  \node (user) {\textbf{User}};
  \node[right=1.0cm of user] (disp) {\textbf{Dispatch}};
  \node[right=0.8cm of disp] (diag) {$\mathbf{A_D}$};
  \node[right=0.8cm of diag] (proc) {$\mathbf{A_P}$};
  \node[right=0.8cm of proc] (synth) {$\mathbf{A_S}$};
  
  \foreach \n in {user,disp,diag,proc,synth}{
    \draw[dashed, gray] (\n) -- ++(0,-5.5);
  }
  
  \draw[->, thick] ([yshift=-0.8cm]user.south) -- node[above, font=\scriptsize] {NL query} ([yshift=-0.8cm]disp.south);
  \draw[->, thick, blue] ([yshift=-1.5cm]disp.south) -- node[above, font=\scriptsize] {\textsf{REQ/TRAVERSE}} ([yshift=-1.5cm]diag.south);
  \draw[->, thick, blue] ([yshift=-2.2cm]diag.south) -- node[above, font=\scriptsize] {\textsf{INFORM/subgraph}} ([yshift=-2.2cm]disp.south);
  \draw[->, thick, blue] ([yshift=-2.9cm]disp.south) -- node[above, font=\scriptsize] {\textsf{REQ/TRAVERSE}} ([yshift=-2.9cm]proc.south);
  \draw[->, thick, blue] ([yshift=-3.6cm]proc.south) -- node[above, font=\scriptsize] {\textsf{QUERY/exists?}} ([yshift=-3.6cm]synth.south);
  \draw[->, thick, blue] ([yshift=-4.3cm]synth.south) -- node[above, font=\scriptsize] {\textsf{CONFIRM/stats}} ([yshift=-4.3cm]proc.south);
  \draw[->, thick, blue] ([yshift=-4.7cm]proc.south) -- node[above, font=\scriptsize] {\textsf{INFORM/procedure}} ([yshift=-4.7cm]disp.south);
  \draw[->, thick] ([yshift=-5.3cm]disp.south) -- node[above, font=\scriptsize] {NL response} ([yshift=-5.3cm]user.south);
  
  \node[left, font=\tiny, gray] at ([yshift=-0.8cm, xshift=-0.2cm]user.south) {NL};
  \node[left, font=\tiny, blue] at ([yshift=-1.5cm, xshift=-0.2cm]user.south) {G²CP};
  \node[left, font=\tiny, blue] at ([yshift=-2.9cm, xshift=-0.2cm]user.south) {G²CP};
  \node[left, font=\tiny, gray] at ([yshift=-5.3cm, xshift=-0.2cm]user.south) {NL};
\end{tikzpicture}
\caption{G²CP message sequence for a diagnostic query. Natural language (black) is used only at the user boundary; all inter-agent communication (blue) uses G²CP graph operations. Labels: REQ = REQUEST.}
\label{fig:sequence}
\end{figure}

\paragraph{Phase 1: Query Decomposition.}
Upon receiving a user query $q$, the Dispatcher uses the LLM to extract entities, classify intent, and select the primary agent (see Section~\ref{sec:llm_selection}). For diagnostic queries, it extracts symptoms $S = \{s_1, \ldots, s_k\}$ and constructs:

\begin{verbatim}
REQUEST TRAVERSE 
  FROM: {s_1, ..., s_k}
  VIA: {causes, indicates}
  DEPTH: 2
  RETURN: SUBGRAPH
\end{verbatim}

\paragraph{Phase 2: Iterative Refinement.}
Agents can REQUEST operations from each other. After $A_D$ identifies fault $f^*$, it sends to $A_P$:

\begin{verbatim}
REQUEST TRAVERSE
  FROM: {f*}
  VIA: {addressed_by}
  DEPTH: 1
  RETURN: SUBGRAPH
\end{verbatim}

\paragraph{Phase 3: Knowledge Integration.}
When $A_S$ discovers a new pattern, it broadcasts:

\begin{verbatim}
UPDATE APPLY
  ADD_EDGE: {
    from: Part:X,
    to: Sensor:temp_anomaly,
    type: risk_indicator,
    confidence: 0.89
  }
\end{verbatim}

\subsubsection{Complete Worked Example}
\label{sec:worked_example}

We trace a full message exchange for the query: \emph{``What causes pressure drops in hydraulic circuit HC-3?''}

\textbf{Step 1.} Dispatcher $\to$ $A_D$:
\begin{verbatim}
Dispatcher TO A_D
PERFORMATIVE: REQUEST
CONVERSATION: conv_042
OPERATION:
  TRAVERSE
    FROM: {Component:HC-3}
    VIA: {has_symptom}
    DEPTH: 1
    RETURN: SUBGRAPH
\end{verbatim}

\textbf{Step 2.} $A_D$ executes the traversal on Neo4j, retrieving symptoms linked to HC-3: \texttt{Symptom:pressure\_drop}, \texttt{Symptom:flow\_reduction}. $A_D$ then performs a second traversal to find root causes:

$A_D$ $\to$ Dispatcher:
\begin{verbatim}
A_D TO Dispatcher
PERFORMATIVE: INFORM
CONVERSATION: conv_042
OPERATION:
  RESULT SUBGRAPH:
    Nodes: {Fault:seal_degradation, 
            Fault:valve_blockage,
            Fault:pump_cavitation}
    Edges: {seal_degradation -[causes]-> 
              pressure_drop,
            valve_blockage -[causes]-> 
              flow_reduction,
            pump_cavitation -[causes]-> 
              pressure_drop}
    Confidence: [0.91, 0.84, 0.72]
\end{verbatim}

\textbf{Step 3.} Dispatcher $\to$ $A_P$ (for top-ranked fault):
\begin{verbatim}
Dispatcher TO A_P
PERFORMATIVE: REQUEST
CONVERSATION: conv_042
OPERATION:
  TRAVERSE
    FROM: {Fault:seal_degradation}
    VIA: {addressed_by, requires_part}
    DEPTH: 1
    RETURN: SUBGRAPH
\end{verbatim}

\textbf{Step 4.} $A_P$ $\to$ $A_S$ (checking historical frequency):
\begin{verbatim}
A_P TO A_S
PERFORMATIVE: QUERY
CONVERSATION: conv_042
OPERATION:
  TRAVERSE
    FROM: {Fault:seal_degradation, 
           Component:HC-3}
    VIA: {occurred_in}
    DEPTH: 2
    RETURN: LEAVES
\end{verbatim}

\textbf{Step 5.} $A_S$ returns historical occurrence data (7 incidents in 18 months). $A_P$ integrates this into the procedure recommendation. Dispatcher generates the final natural language response from the collected subgraphs.

\textbf{Total tokens}: 189 (vs.\ estimated 1,400+ for equivalent FTMA exchange).

\subsection{Implementation Details}

Each agent maintains a local graph cache (a materialized view of $G$ optimized for its edge specialization), a message queue (pending G²CP messages with priority ordering), an execution engine (Cypher-based traversal executor), and an LLM interface (used only for parsing user queries and formatting final responses, not for inter-agent communication).

The system uses Neo4j for graph storage with custom traversal APIs. Message passing occurs through Apache Kafka for reliability and auditability.

\subsection{G²CP Runtime Engine}

We detail the implementation architecture that executes G²CP messages.

\paragraph{Message Parser.}
Each agent runs a G²CP parser implementing a three-stage pipeline:

\begin{algorithm}[h]
\caption{G²CP Message Parsing and Execution}
\label{alg:parsing}
\small
\begin{algorithmic}[1]
\STATE \textbf{Input:} Raw message $m_{raw}$
\STATE \textbf{Output:} Execution result $R$ or error
\STATE // Stage 1: Parse and validate syntax
\STATE $m \gets \text{ParseMessage}(m_{raw})$
\IF{$m = \text{INVALID}$}
    \RETURN $\text{ERROR}(\text{``Malformed message''})$
\ENDIF
\STATE // Stage 2: Security validation
\IF{$\neg \text{AuthorizeAgent}(m.\mathit{sender}, m.\mathit{op})$}
    \RETURN $\text{ERROR}(\text{``Unauthorized operation''})$
\ENDIF
\STATE // Stage 3: Execute operation
\IF{$m.\mathit{op}$ is TRAVERSE}
    \STATE $V_s \gets \text{ResolveNodeSelector}(m.\mathit{op}.from)$
    \STATE $R \gets \text{ExecuteTraversal}(V_s, m.\mathit{op}.via, m.\mathit{op}.depth)$
\ELSIF{$m.\mathit{op}$ is UPDATE}
    \STATE $R \gets \text{ApplyUpdate}(G, m.\mathit{op}.\Delta G)$
\ENDIF
\STATE // Stage 4: Log for auditability
\STATE $\text{AuditLog}.\text{append}(m, R, \text{timestamp})$
\RETURN $R$
\end{algorithmic}
\end{algorithm}

\paragraph{Node Resolution.}
Node selectors use a priority system: (1)~Explicit IDs via indexed lookup, (2)~Type filters via Cypher \texttt{MATCH (n:Type) RETURN n}, (3)~Property filters compiled to parameterized Cypher with sanitization, (4)~Context references resolved from conversation state.

\paragraph{Traversal Executor.}
Graph traversals use Neo4j's pattern matching:

\begin{verbatim}
MATCH path = (start)-[r*1..{depth}]->(end)
WHERE start IN {node_set}
  AND type(r) IN {edge_filter}
RETURN path
\end{verbatim}

For large-scale traversals ($h > 2$), breadth-first frontier expansion with early termination at 1000 nodes is employed.

\paragraph{Error Handling.}
The executor implements timeout protection (30s), result size limits (5000 nodes), malformed operation reporting, and graph inconsistency handling with retry or human escalation.

\subsection{LLM-Based Operation Selection}
\label{sec:llm_selection}

A key clarification: agents are \textbf{not pre-scripted}. Each agent uses LLM-based reasoning to dynamically select graph operations at runtime. Only the graph schema (node and edge types) is pre-defined---identical across all baselines. Algorithm~\ref{alg:llm_selection} details the selection process.

\begin{algorithm}[t]
\caption{LLM-Based Graph Operation Selection}
\label{alg:llm_selection}
\small
\begin{algorithmic}[1]
\STATE \textbf{Input:} User query $q$ or upstream G²CP message $m$
\STATE \textbf{Output:} G²CP operation $\mathit{op}$
\STATE
\STATE // Step 1: Entity extraction via LLM
\STATE $\mathit{entities} \gets \text{LLM}(\text{EXTRACT\_PROMPT}, q)$
\STATE \COMMENT{e.g., $\{$``bearing B-4521'', ``grinding noise''$\}$}
\STATE
\STATE // Step 2: Entity linking to graph nodes
\STATE $V_s \gets \emptyset$
\FOR{each $e \in \mathit{entities}$}
    \STATE $v \gets \text{FuzzyMatch}(e, V, \text{threshold}=0.85)$
    \IF{$v \neq \text{NULL}$}
        \STATE $V_s \gets V_s \cup \{v\}$
    \ENDIF
\ENDFOR
\STATE
\STATE // Step 3: Intent classification
\STATE $\mathit{intent} \gets \text{LLM}(\text{INTENT\_PROMPT}, q)$
\STATE \COMMENT{One of: diagnostic, procedural, predictive, factoid}
\STATE
\STATE // Step 4: Edge type selection based on intent
\STATE $\Psi_f \gets \text{EDGE\_MAP}[\mathit{intent}]$
\STATE \COMMENT{e.g., diagnostic $\to$ \{causes, indicates\}}
\STATE
\STATE // Step 5: Depth estimation
\STATE $h \gets \text{LLM}(\text{DEPTH\_PROMPT}, q, |V_s|, |\Psi_f|)$
\STATE \COMMENT{Typically 1--3 based on query complexity}
\STATE
\STATE // Step 6: Assemble G²CP operation
\STATE $\mathit{op} \gets \mathtt{TRAVERSE}(V_s, \Psi_f, h, \mathit{ret})$
\RETURN $\mathit{op}$
\end{algorithmic}
\end{algorithm}

\paragraph{Example Prompts.}
The entity extraction prompt used by the Dispatcher:
\begin{quote}
\small\texttt{Given the query: "\{query\}"\\
Extract all named entities that correspond to industrial equipment, symptoms, faults, or parts. Return as JSON: \{"entities": [...]\}}
\end{quote}

The intent classification prompt:
\begin{quote}
\small\texttt{Classify the query intent into one of: diagnostic (symptom$\to$cause), procedural (fault$\to$fix), predictive (pattern$\to$forecast), factoid (direct lookup). Query: "\{query\}" Intent:}
\end{quote}

The depth estimation prompt:
\begin{quote}
\small\texttt{Given \{n\_entities\} source entities and \{n\_edge\_types\} edge types, estimate the number of traversal hops (1-3) needed. Simple lookups: 1. Causal chains: 2. Complex multi-factor analysis: 3. Query: "\{query\}" Depth:}
\end{quote}


\section{Experimental Evaluation}
\label{sec:experiments}

\subsection{Experimental Setup}

\subsubsection{Knowledge Graph}
We constructed a synthetic industrial knowledge base for a ``Turbomatic MKII Hydraulic Press'' containing 1,247 nodes across 7 types (Components, Faults, Procedures, Work Orders, Parts, Sensors, Safety Protocols), 3,892 edges across 12 relationship types, with graph density 0.005, average degree 6.24, and diameter 8.

\subsubsection{Knowledge Graph Construction Methodology}
\label{sec:kg_construction}

The knowledge graph was constructed through a four-stage pipeline:

\begin{enumerate}
\item \textbf{Manual extraction} (Stage 1): Two domain engineers extracted 89 components and 45 fault types from the Turbomatic MKII technical manual (312 pages), identifying part hierarchies, known failure modes, and safety protocols.

\item \textbf{Work order integration} (Stage 2): We generated 200 synthetic work orders modeled on real maintenance logs from our industrial partners, covering a 3-year operational period. Each work order links symptoms, diagnosed faults, applied procedures, and replaced parts.

\item \textbf{Relationship inference} (Stage 3): Co-occurrence analysis over work orders produced 1,200+ candidate edges (e.g., if ``seal degradation'' and ``pressure drop'' co-occur in $>$60\% of relevant work orders, an \texttt{indicates} edge is created).

\item \textbf{Expert validation} (Stage 4): Three maintenance technicians reviewed all extracted entities and relationships, correcting 127 errors and adding 89 missing edges. Inter-annotator agreement (Fleiss' $\kappa$) was 0.81.
\end{enumerate}

\subsubsection{Representative Queries}
\label{sec:representative_queries}

Table~\ref{tab:queries} presents representative queries from each category with difficulty metrics.

\begin{table*}[t]
\centering
\caption{Representative queries across categories with difficulty metrics and expected behavior}
\label{tab:queries}
\small
\begin{tabular}{@{}clccp{4.5cm}@{}}
\toprule
\textbf{Cat.} & \textbf{Example Query} & \textbf{Hops} & \textbf{Entities} & \textbf{Failure Mode (FTMA)} \\
\midrule
\multirow{2}{*}{Fact.} & What is the rated pressure of pump P-101? & 1 & 1 & Retrieves wrong pump \\
& List all sensors on hydraulic circuit HC-3 & 1 & 1 & Incomplete listing \\
\midrule
\multirow{2}{*}{Diag.} & Why is there grinding noise at 1200 RPM? & 2 & 1 & Ambiguous fault identification \\
& Diagnose: pressure drop + high temp in circuit 7 & 2 & 3 & Wrong causal chain \\
\midrule
\multirow{2}{*}{Proc.} & How do I replace bearing B-4521? & 1 & 1 & Wrong bearing procedure \\
& Full procedure for seal kit replacement on HP-2 & 2 & 2 & Missing safety steps \\
\midrule
\multirow{2}{*}{Rel.} & Which faults affect both circuit HC-3 and HC-7? & 2 & 2 & Misses shared components \\
& What parts are needed for all pending work orders? & 3 & $5+$ & Incomplete aggregation \\
\midrule
\multirow{2}{*}{Pred.} & Predict next likely failure for pump P-101 & 3 & 1 & Hallucinated prediction \\
& Which components are at risk given sensor trends? & 3 & $3+$ & No historical grounding \\
\bottomrule
\end{tabular}
\end{table*}

\subsubsection{Baseline Implementation Details}
\label{sec:baseline_details}

All four systems share identical infrastructure: Neo4j graph database, GPT-4 for query understanding, Llama 3 70B for response generation, and access to the same knowledge graph via Cypher queries. The \textbf{only difference} is the inter-agent communication format. Below we provide the core system prompts for each baseline.

\paragraph{FTMA (Free-Text Multi-Agent).}
Agents communicate in natural language. Diagnostic Agent system prompt:
\begin{quote}
\small\texttt{You are a diagnostic specialist. You have access to a graph\_query tool that executes Cypher queries on the industrial knowledge graph. When another agent asks you to diagnose a problem, analyze the symptoms, query the graph for relevant fault information, and respond in natural language with your findings.}
\end{quote}

\paragraph{JSMA (JSON-Structured Multi-Agent).}
Agents exchange JSON objects. Diagnostic Agent system prompt:
\begin{quote}
\small\texttt{You are a diagnostic specialist. You have access to a graph\_query tool that executes Cypher queries. When receiving a JSON message like \{"action": "diagnose", "symptoms": [...]\}, query the graph and respond with \{"faults": [...], "confidence": [...]\}.}
\end{quote}

\paragraph{G²CP.}
Diagnostic Agent system prompt:
\begin{quote}
\small\texttt{You are a diagnostic specialist. You receive G²CP messages containing TRAVERSE operations over the knowledge graph. Execute the specified graph operation exactly as described. Return results as an INFORM message with the retrieved subgraph. Select edge types from your specialization: \{causes, indicates, correlates\_with\}.}
\end{quote}

\paragraph{Single-Agent Baseline.}
A monolithic RAG system with no agent decomposition, using the same graph and LLMs with a single prompt combining all agent capabilities.

Complete prompts for all four agent roles across all systems are provided in Appendix~\ref{app:prompts}.

\paragraph{Evaluation Metrics.}
\textbf{Task Completion Accuracy}: fraction of queries answered correctly (exact match on ground truth). \textbf{Token Efficiency}: total tokens consumed across all inter-agent messages (excluding internal reasoning; see Section~\ref{sec:token_counting}). \textbf{Hallucination Rate}: percentage of claims unsupported by graph evidence. \textbf{Cascading Error Rate}: fraction of failures caused by miscommunication between agents. \textbf{Auditability Score}: percentage of reasoning steps verifiable from message logs.

\subsection{Overall Performance}

Table~\ref{tab:overall} presents aggregate results across all 500 queries.

\begin{table}[t]
\centering
\caption{Overall system performance comparison. Bold indicates best. All G²CP improvements over FTMA are statistically significant ($p < 0.001$, paired $t$-test).}
\label{tab:overall}
\small
\begin{tabular}{@{}lcccc@{}}
\toprule
\textbf{Metric} & \textbf{FTMA} & \textbf{JSMA} & \textbf{Single} & \textbf{G²CP} \\
\midrule
Task Accuracy & 0.67 & 0.74 & 0.71 & \textbf{0.90} \\
Token Efficiency & 2,847 & 2,134 & 1,456 & \textbf{768} \\
Hallucination Rate & 0.23 & 0.18 & 0.14 & \textbf{0.02} \\
Cascading Errors & 0.31 & 0.19 & 0.00 & \textbf{0.00} \\
Auditability & 0.42 & 0.68 & 1.00 & \textbf{1.00} \\
Avg Response Time (s) & 4.2 & 3.8 & 2.1 & 2.9 \\
\bottomrule
\end{tabular}
\end{table}

\paragraph{Accuracy Improvement.}
G²CP achieves 90\% task completion accuracy, a 34\% relative improvement over FTMA (67\%). The improvement is most pronounced on relational and predictive queries requiring multi-hop reasoning (Figure~\ref{fig:accuracy_by_category}).

\paragraph{Token Efficiency.}
G²CP consumes 768 tokens per query on average, a 73\% reduction compared to FTMA (2,847 tokens). This stems from replacing verbose natural language exchanges with compact graph operations. For example:

\textbf{FTMA exchange} (287 tokens):
\begin{quote}
\small\textit{Agent A: ``I've identified that the grinding noise at 1200 RPM combined with pressure fluctuations suggests bearing wear in the main pump assembly. Can you help me find the appropriate repair procedure?''\\
Agent B: ``Thank you for the diagnosis. I'll search for repair procedures related to bearing replacement in the main pump assembly\ldots''}
\end{quote}

\textbf{G²CP exchange} (41 tokens):
\begin{verbatim}
REQUEST TRAVERSE FROM {Fault:bearing_wear_B4521}
VIA {addressed_by} DEPTH 1 RETURN SUBGRAPH
\end{verbatim}

\paragraph{Hallucination Elimination.}
G²CP reduces hallucinations to 2\%, compared to 23\% in FTMA. G²CP eliminates fabrication through deterministic graph operations---agents cannot ``make up'' edges that do not exist in $G$.

\paragraph{Cascading Error Elimination.}
FTMA exhibits 31\% cascading error rate. G²CP's explicit node references eliminate referential ambiguity entirely.

\subsection{Performance by Query Category}

Figure~\ref{fig:accuracy_by_category} shows accuracy breakdown by query type.

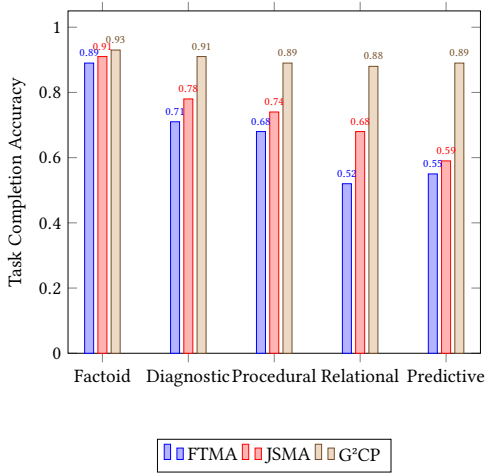
\begin{figure}[t]
\centering
\begin{tikzpicture}[scale=0.8]
  \begin{axis}[
    ybar,
    bar width=0.15cm,
    ylabel={Task Completion Accuracy},
    symbolic x coords={Factoid, Diagnostic, Procedural, Relational, Predictive},
    xtick=data,
    legend style={at={(0.5,-0.25)}, anchor=north, legend columns=3},
    ymin=0, ymax=1.05,
    nodes near coords,
    every node near coord/.append style={font=\tiny},
  ]
    \addplot coordinates {(Factoid,0.89) (Diagnostic,0.71) (Procedural,0.68) (Relational,0.52) (Predictive,0.55)};
    \addplot coordinates {(Factoid,0.91) (Diagnostic,0.78) (Procedural,0.74) (Relational,0.68) (Predictive,0.59)};
    \addplot coordinates {(Factoid,0.93) (Diagnostic,0.91) (Procedural,0.89) (Relational,0.88) (Predictive,0.89)};
    \legend{FTMA, JSMA, G²CP}
  \end{axis}
\end{tikzpicture}
\caption{Task completion accuracy by query category}
\label{fig:accuracy_by_category}
\end{figure}

G²CP shows largest improvements on relational (69\% relative gain) and predictive queries (62\% gain), where multi-hop reasoning and pattern discovery are essential.

\subsection{Ablation Studies}

Table~\ref{tab:ablation} analyzes the contribution of G²CP components.

\begin{table}[t]
\centering
\caption{Ablation study: impact of removing G²CP features}
\label{tab:ablation}
\small
\begin{tabular}{@{}lcc@{}}
\toprule
\textbf{Configuration} & \textbf{Accuracy} & \textbf{Token Usage} \\
\midrule
Full G²CP & \textbf{0.90} & \textbf{768} \\
\midrule
$-$ Explicit node IDs & 0.82 & 891 \\
$-$ Edge type constraints & 0.78 & 823 \\
$-$ Hop depth specification & 0.74 & 1,234 \\
$-$ Subgraph return format & 0.85 & 945 \\
$-$ Context tracking & 0.81 & 812 \\
\bottomrule
\end{tabular}
\end{table}

Removing hop depth specification causes the largest degradation (18\% accuracy drop), as agents default to exhaustive traversal, retrieving irrelevant context.

\subsection{Case Study: Complex Diagnostic Query}
\label{sec:case_study}

We analyze the query: \textit{``The hydraulic press is making grinding noise at 1200 RPM with pressure fluctuations and 85°C oil temperature. What's wrong and how do I fix it?''}

\paragraph{FTMA Execution (Failed).}
Dispatcher $\to$ $A_D$: ``Analyze symptoms: grinding noise, pressure issues, high temperature.'' $A_D$ $\to$ Dispatcher: ``Possible causes include pump bearing wear or cavitation.'' Dispatcher $\to$ $A_P$: ``Find procedures for pump bearing or cavitation.'' $A_P$ retrieves two procedures, unsure which applies. Final response: generic troubleshooting (incorrect). Total tokens: 3,124.

\paragraph{G²CP Execution (Success).}
We show the complete knowledge access pipeline:

\textbf{Step 1: Entity extraction.} LLM extracts: \texttt{\{grinding\_1200RPM, pressure\_fluctuation, temp\_85C\}}.

\textbf{Step 2: Entity linking.} Fuzzy matching links to graph nodes: \texttt{Symptom:grinding\_1200RPM} (score 0.97), \texttt{Symptom:pressure\_fluctuation} (0.94), \texttt{Symptom:temp\_85C} (0.99).

\textbf{Step 3: Operation selection.} Intent classified as ``diagnostic'' $\to$ edge types \texttt{\{causes, indicates\}}, depth estimated as 2.

\textbf{Step 4: Traversal execution.}
\begin{verbatim}
REQUEST TRAVERSE
  FROM {Symptom:grinding_1200RPM, 
        Symptom:pressure_fluctuation,
        Symptom:temp_85C}
  VIA {causes, indicates}
  DEPTH 2
  RETURN PATHS
\end{verbatim}

\textbf{Step 5: Result subgraph.} The traversal returns convergent paths:
\begin{verbatim}
grinding_1200RPM -[indicates]-> bearing_wear
                 -[located_in]-> B-4521
temp_85C -[causes]-> lubrication_failure 
         -[leads_to]-> bearing_wear
pressure_fluctuation -[indicates]-> bearing_wear
\end{verbatim}

All three paths converge on \texttt{Fault:bearing\_wear\_B4521} (confidence 0.94).

\textbf{Step 6: Procedure retrieval.}
\begin{verbatim}
REQUEST TRAVERSE
  FROM {Fault:bearing_wear_B4521}
  VIA {addressed_by, requires_part, 
       has_safety_protocol}
  DEPTH 1
  RETURN SUBGRAPH
\end{verbatim}

Returns procedure P-205 with parts list and safety protocols. Total tokens: 623. \textbf{Every claim in the final response maps to specific graph paths.}

\subsubsection{Query-to-G²CP Translation Pipeline}
\label{sec:translation_pipeline}

Algorithm~\ref{alg:translation} formalizes the complete pipeline from natural language query to G²CP operations.

\begin{algorithm}[t]
\caption{Query-to-G²CP Translation}
\label{alg:translation}
\small
\begin{algorithmic}[1]
\STATE \textbf{Input:} Natural language query $q$
\STATE \textbf{Output:} Sequence of G²CP messages $\langle m_1, \ldots, m_k \rangle$
\STATE \textit{// Entity extraction via LLM}
\STATE $E_{\text{raw}} \gets \text{LLM}(\text{EXTRACT}, q)$
\STATE \textit{// Entity linking to graph nodes}
\STATE $V_s \gets \emptyset$
\FOR{each $e \in E_{\text{raw}}$}
    \STATE $\mathbf{x}_e \gets \text{Encode}(e)$ \COMMENT{sentence embedding}
    \STATE $v^* \gets \arg\max_{v \in V} \cos(\mathbf{x}_e, \mathbf{x}_v)$
    \STATE \textit{// threshold $\tau = 0.85$}
    \IF{$\cos(\mathbf{x}_e, \mathbf{x}_{v^*}) > \tau$}
        \STATE $V_s \gets V_s \cup \{v^*\}$
    \ENDIF
\ENDFOR
\STATE \textit{// Intent classification}
\STATE $\mathit{intent} \gets \text{LLM}(\text{CLASSIFY}, q)$
\STATE \textit{// Edge selection based on intent}
\STATE $\Psi_f \gets \text{EDGE\_MAP}[\mathit{intent}]$
\STATE \textit{// Depth estimation based on complexity}
\STATE $h \gets \text{LLM}(\text{DEPTH}, q, |V_s|)$
\STATE \textit{// Generate primary operation}
\STATE $m_1 \gets \langle \text{Dispatch}, A_{\mathit{intent}}, \text{REQUEST}, $
\STATE $\quad\quad\quad \mathtt{TRAVERSE}(V_s, \Psi_f, h, \text{SUBGRAPH}), \mathit{ctx} \rangle$
\STATE \textit{// Downstream operations generated iteratively based on results of $m_1$}
\RETURN $\langle m_1, \ldots \rangle$
\end{algorithmic}
\end{algorithm}

\subsection{Scalability Analysis}

Figure~\ref{fig:scalability} shows response time scaling with graph size.

\begin{figure}[t]
\centering
\begin{tikzpicture}[scale=0.9]
  \begin{axis}[
    xlabel={Knowledge Graph Size (nodes)},
    ylabel={Response Time (seconds)},
    legend pos=north west,
    xmode=log,
    ymin=0,
    grid=major,
  ]
    \addplot[mark=square, blue, thick] coordinates {
      (1000,1.2) (5000,2.1) (10000,3.4) (50000,4.8) (100000,9.2)
    };
    \addplot[mark=triangle, red, thick] coordinates {
      (1000,1.8) (5000,3.9) (10000,8.1) (50000,21.3) (100000,47.8)
    };
    \legend{G²CP, FTMA}
  \end{axis}
\end{tikzpicture}
\caption{Response time scaling with knowledge graph size. G²CP maintains sub-linear scaling ($O(n^{0.7})$); FTMA scales super-linearly ($O(n^{1.3})$).}
\label{fig:scalability}
\end{figure}
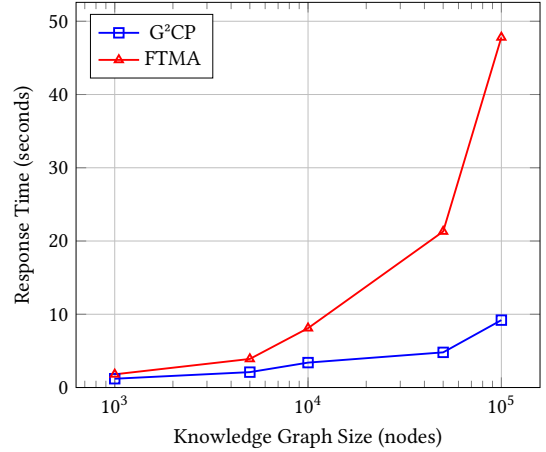

At 100,000 nodes, G²CP is 5.2$\times$ faster than FTMA.

\subsection{Human Evaluation}

We conducted a user study with industrial maintenance technicians who evaluated system outputs on clarity, completeness, and trustworthiness (1--5 Likert scale):

\begin{table}[h]
\centering
\small
\begin{tabular}{@{}lccc@{}}
\toprule
\textbf{System} & \textbf{Clarity} & \textbf{Completeness} & \textbf{Trust} \\
\midrule
FTMA & 3.2 & 2.9 & 2.7 \\
G²CP & 4.1 & 4.3 & 4.5 \\
\bottomrule
\end{tabular}
\end{table}

Technicians particularly valued G²CP's auditability and specificity (exact part numbers and procedures rather than generic advice).

\subsection{Real-World Industrial Validation}

To validate beyond synthetic benchmarks, we conducted a pilot deployment with two industrial partners: a hydraulic equipment manufacturer and an automotive parts supplier. We collected 25 real diagnostic cases from their maintenance logs spanning 9 months of operations.

\paragraph{Dataset Characteristics.}
Real cases differed from synthetic data: 68\% had missing sensor data or ambiguous symptoms, 32\% involved simultaneous multi-component failures, 44\% required cross-referencing work orders from 6+ months prior, and technicians used colloquial terminology requiring interpretation.

\paragraph{Knowledge Graph Construction.}
Partner A (Hydraulics): 847 nodes, 2,341 edges covering 12 equipment types. Partner B (Automotive): 1,124 nodes, 3,156 edges covering 8 production lines. Graphs integrated technical manuals, 200+ historical work orders, equipment specifications, and safety protocols.

\paragraph{Evaluation Protocol.}
Each of the 25 cases was evaluated by three maintenance technicians (domain experts) who provided ground truth. Expert consensus required 2/3 agreement; 4 cases without consensus were excluded, leaving 21 evaluable cases.

\paragraph{Results.}

\begin{table}[t]
\centering
\caption{Performance on real industrial maintenance cases ($n=21$)}
\label{tab:realworld}
\small
\begin{tabular}{@{}lcc@{}}
\toprule
\textbf{Metric} & \textbf{FTMA} & \textbf{G²CP} \\
\midrule
Correct Diagnosis & 11/21 (52\%) & 18/21 (86\%) \\
Correct Procedure & 9/21 (43\%) & 17/21 (81\%) \\
Complete Resolution & 7/21 (33\%) & 15/21 (71\%) \\
Avg.\ Response Time (s) & 6.8 & 4.2 \\
Technician Approval & 48\% & 89\% \\
\bottomrule
\end{tabular}
\end{table}

G²CP achieved 86\% diagnostic accuracy on real cases versus 90\% on synthetic data, demonstrating robustness to real-world noise.

\paragraph{Qualitative Findings.}
Technicians reported trusting G²CP because they could verify graph paths. G²CP surfaced 4 historical failure patterns that current technicians were unaware of. Average time-to-diagnosis decreased from 47 minutes (manual) to 12 minutes (G²CP-assisted).

\paragraph{Failure Analysis.}
The 3 failures: (1)~a custom-modified component not in the knowledge graph, (2)~``intermittent noise'' too ambiguous for reliable linking, (3)~a previously unseen failure mode requiring human expert diagnosis. These highlight the importance of continuous graph maintenance.

\balance


\section{Theoretical Analysis}
\label{sec:theory}

We present key theoretical results; detailed proofs are in Appendix~\ref{app:proofs}.

\subsection{Computational Complexity}

\begin{theorem}[Traversal Complexity]
A traversal $\mathtt{TRAVERSE}\\(V_s, \Psi_f, h, \mathit{ret})$ has time complexity $O(|V_s| \cdot d^h \cdot |\Psi_f|)$ where $d$ is the average node degree.
\end{theorem}

\emph{Proof sketch.} Each hop expands the frontier by at most $d \cdot |\Psi_f|$ nodes. With $h$ hops from $|V_s|$ sources, the explored region is $O(|V_s| \cdot d^h \cdot |\Psi_f|)$. In practice, $h \leq 3$ and $|\Psi_f| \ll |\Psi|$, making traversals highly efficient. Full proof in Appendix~\ref{app:proofs}. $\square$

\subsection{Message Complexity}

\begin{theorem}[Communication Efficiency]
For a task requiring $k$ agent interactions, G²CP consumes $O(k \cdot \log |V|)$ tokens, while free-text communication consumes $O(k \cdot L)$ tokens where $L$ is the average natural language message length.
\end{theorem}

\emph{Proof sketch.} G²CP messages encode node IDs ($O(\log |V|)$ each), constant-size edge types, and small integer hop counts. Natural language messages require $L \approx 100$--$300$ tokens. Since $\log |V| \ll L$, G²CP achieves substantial savings. Full proof in Appendix~\ref{app:proofs}. $\square$

\subsubsection{Token Counting Methodology}
\label{sec:token_counting}

We count \textbf{all inter-agent messages} (the content exchanged between agents), excluding: (1)~internal LLM reasoning within a single agent, (2)~the original user query, and (3)~the final response to the user. This isolates the communication overhead that G²CP specifically targets.

Token counts are computed using the \texttt{tiktoken} cl100k\_base tokenizer (consistent with GPT-4). The empirical range of 100--300 tokens per FTMA message reflects the distribution across our 500 queries:

\begin{table}[h]
\centering
\small
\begin{tabular}{@{}lcc@{}}
\toprule
\textbf{Query Complexity} & \textbf{FTMA tokens/msg} & \textbf{G²CP tokens/msg} \\
\midrule
Factoid (1-hop) & $87 \pm 23$ & $28 \pm 6$ \\
Diagnostic (2-hop) & $194 \pm 41$ & $38 \pm 8$ \\
Procedural (2-hop) & $178 \pm 38$ & $35 \pm 7$ \\
Relational (2-3 hop) & $267 \pm 52$ & $45 \pm 11$ \\
Predictive (3-hop) & $312 \pm 67$ & $52 \pm 14$ \\
\bottomrule
\end{tabular}
\end{table}

The token reduction scales with query complexity: simpler queries see $\sim$3$\times$ reduction, while complex multi-hop queries see $\sim$6$\times$ reduction.

\subsection{Correctness Guarantees}

\begin{theorem}[Non-Hallucination]
If agent $A$ makes a claim $c$ with G²CP provenance trace $\tau = \langle m_1, \ldots, m_k \rangle$, then either $c$ is grounded in $G$ or $\tau$ is invalid (detectable through replay).
\end{theorem}

\emph{Proof sketch.} Replaying $\tau$ from $G$ produces a unique subgraph $G_\tau$. If $c$ is fabricated, either it fails verification against $G_\tau$ or the trace is falsified (replay produces different results). Both cases are detectable. Full proof in Appendix~\ref{app:proofs}. $\square$

\subsection{Security and Trust Model}

G²CP operates in multi-agent environments where malicious or compromised agents pose risks.

\paragraph{Threat Model.}
We consider: (1)~malicious agents sending harmful operations, (2)~compromised legitimate agents, and (3)~man-in-the-middle message modification.

\paragraph{Agent Authentication.}
Each agent $A_i \in \mathcal{A}$ has Ed25519 key pair $(pk_i, sk_i)$. All messages are signed: $m_{\text{signed}} = \langle m, \sigma \rangle$ where $\sigma = \text{Sign}(sk_{\text{sender}}, \text{Hash}(m))$.

\paragraph{Access Control.}
Role-based access control (RBAC) where each agent has permissions $\mathcal{P}(A_i) \subseteq \{\text{READ}, \text{TRAVERSE}, \text{UPDATE}\} \times 2^{\Lambda} \times 2^{\Psi}$.\\ For example: Diagnostic Agent has $(\text{TRAVERSE},\\ \{\text{Symptom, Fault}\}, \{\text{causes, indicates}\})$; Ingestion Agent has $(\text{UPDATE}, \Lambda, \Psi)$.

\paragraph{Graph Integrity.}
UPDATE operations undergo type constraint validation, relationship schema enforcement, provenance tagging, and rollback-capable graph versioning.

\paragraph{Trust Propagation.}
Agents maintain trust scores $\tau_{t+1}(A_j) = \alpha \cdot \tau_t(A_j) + (1-\alpha) \cdot \mathbb{I}[\text{Verify}(A_j, m_t)]$ with $\alpha = 0.9$. Low-trust agents ($\tau < 0.5$) trigger human review.


\section{Discussion}

\subsection{Advantages of G²CP}

\paragraph{Precision.} Graph operations eliminate referential ambiguity. \texttt{FROM \{Part:B-4521\}} is unambiguous, unlike ``the main bearing.''

\paragraph{Efficiency.} 73\% token reduction translates to lower API costs and faster inference.

\paragraph{Auditability.} Every reasoning step is verifiable by replaying graph operations---essential for healthcare, finance, and industrial safety.

\paragraph{Composability.} Complex queries decompose naturally into operation sequences without semantic drift.

\paragraph{Scalability.} G²CP scales sublinearly with graph size, while natural language approaches degrade due to context length limitations.

\subsection{Limitations}

G²CP requires a structured knowledge graph; for domains without existing graphs, construction requires 200--500 hours. If required information is absent from the graph, G²CP cannot help, whereas free-text agents can sometimes leverage LLM world knowledge (at the risk of hallucination). Developers must understand graph query semantics, though we argue this is no more complex than SQL or SPARQL. Real-time data (e.g., live sensor readings) may not be in the graph; hybrid approaches with external API calls are needed.

\subsection{Generalizability}

While validated in industrial contexts, G²CP applies to any domain with structured knowledge: healthcare (symptoms $\to$ diseases $\to$ treatments), legal (precedents $\to$ statutes $\to$ arguments), scientific research (citation networks, experimental results), and software development (code dependency graphs). The key requirement is that domain knowledge can be represented as a graph with meaningful node and edge types.

\subsection{Future Directions}

Promising extensions include: reinforcement learning for traversal strategy optimization, federated deployment across organizational boundaries with privacy constraints, temporal logic operators for time-evolving graphs, uncertainty quantification with probabilistic edges, and interactive natural language interfaces for inspecting agent-proposed traversals.

\section{Conclusion}

We presented G²CP, a graph-grounded communication protocol that replaces natural language with structured graph operations for multi-agent coordination. Through formal analysis, evaluation on 500 synthetic and 21 real-world industrial scenarios, we demonstrated 34\% accuracy improvement over free-text baselines, 73\% token reduction, complete elimination of cascading errors and hallucination propagation, and full auditability with deterministic reasoning traces.

G²CP bridges classical agent communication languages with modern neural architectures: it preserves the structured performatives and social commitments of FIPA-ACL while grounding content in graph operations compatible with LLM-based systems. The protocol's impact extends beyond industrial applications to any domain requiring precise, verifiable agent coordination.

We release our G²CP specification, implementation, evaluation datasets, and all baseline code to foster adoption and research.\footnote{Repository: \url{https://github.com/<<anonymous>>/g2cp}}

\begin{acks}
This research was supported by the University of Lorraine and the ERPI laboratory. We thank the industrial partners who provided domain expertise and evaluation support.
\end{acks}


\bibliographystyle{ACM-Reference-Format}


\appendix

\section{Detailed Proofs}
\label{app:proofs}

\begin{proof}[Proof of Theorem~\ref{thm:determinism} (Determinism)]
Traversal operations are defined by recursive set expansion with deterministic neighborhood functions. Specifically, $N(V_s, \Psi_f)$ is computed by iterating over all edges incident to $V_s$ and filtering by type---a deterministic set operation. The base case ($h=0$) returns a fixed set, and the recursive case applies the same deterministic expansion. Update operations are set-theoretic unions and differences. Neither traversal nor update involves stochastic processes, random sampling, or agent-specific interpretation. Therefore, for any fixed $G$ and $op$, the result is unique.
\end{proof}

\begin{proof}[Proof of Theorem~\ref{thm:auditability} (Auditability)]
Each message $m_i$ specifies explicit graph operations. Given initial state $G_0$ and message sequence $\langle m_1, \ldots, m_k \rangle$, we compute:
\begin{equation}
G_i = \begin{cases}
G_{i-1} & \text{if } m_i \text{ is traversal} \\
\mathtt{UPDATE}(G_{i-1}, \Delta G_i) & \text{if } m_i \text{ is update}
\end{cases}
\end{equation}
Each traversal on $G_i$ produces a deterministic subgraph (by Theorem~\ref{thm:determinism}). An auditor can verify that each agent's output is consistent with the subgraphs retrieved at the corresponding graph state. If any inconsistency is found, the audit fails, identifying the offending agent and operation.
\end{proof}

\begin{proof}[Proof of Theorem~\ref{thm:completeness} (Completeness)]
Any graph-based reasoning task decomposes into: (1)~identifying relevant entities (node selection via $V_s$), (2)~exploring relationships (edge traversal via $\Psi_f$ and $h$), (3)~extracting context (subgraph retrieval via $\mathit{ret}$). These map directly to $\mathtt{TRAVERSE}(V_s, \Psi_f, h, \mathit{ret})$. Complex queries requiring multiple perspectives compose through message sequences: Agent $A$ retrieves subgraph $G_1$, passes relevant nodes to Agent $B$, who retrieves $G_2$, etc. The union $G_1 \cup G_2 \cup \ldots$ covers any query answerable from $G$. Graph updates are handled by $\mathtt{UPDATE}$ operations. Therefore, G²CP can express any query answerable through graph traversal and RAG.
\end{proof}

\begin{proof}[Proof of Traversal Complexity]
At hop 0, the frontier is $V_s$ with $|V_s|$ nodes. At each hop $i \in [1, h]$, each frontier node has at most $d$ neighbors, of which at most $|\Psi_f|$ edge types are checked (each in $O(1)$ with hash-based lookup). The frontier grows to at most $|V_s| \cdot (d \cdot |\Psi_f|)^h$ nodes in the worst case. Since each node and edge is processed once, total complexity is $O(|V_s| \cdot d^h \cdot |\Psi_f|)$.
\end{proof}

\begin{proof}[Proof of Communication Efficiency]
A G²CP message contains: sender/receiver IDs (constant), performative (constant), node IDs (each requiring $\lceil \log_2 |V| \rceil$ bits, or $O(\log |V|)$ tokens for a set of bounded size), edge type identifiers (constant per type, bounded set), hop count (single integer), and return format (constant). Total: $O(\log |V|)$ tokens.

A natural language message describes entities by name/description ($O(L_e)$ tokens per entity), relationships in prose ($O(L_r)$ tokens), and intent/context ($O(L_c)$ tokens). Empirically, $L = L_e + L_r + L_c \approx 100$--$300$ tokens.

Over $k$ interactions: G²CP uses $O(k \cdot \log |V|)$ tokens; free text uses $O(k \cdot L)$ tokens. For $|V| \leq 10^6$, $\log |V| \leq 20 \ll 100 \leq L$.
\end{proof}

\begin{proof}[Proof of Non-Hallucination]
Given claim $c$ with provenance trace $\tau = \langle m_1, \ldots, m_k \rangle$:

Replay $\tau$ from $G_0$ to obtain subgraph $G_\tau = \bigcup_{i} \mathcal{T}(m_i.\mathit{op}, G_{i-1})$.

Case 1: $c$ is derivable from $G_\tau$. Then $c$ is grounded---verified.

Case 2: $c$ is not derivable from $G_\tau$. Either (a)~agent fabricated $c$ without graph support, in which case $c$ fails verification against $G_\tau$ (the auditor checks that every entity and relationship in $c$ appears in $G_\tau$), or (b)~the trace $\tau$ is falsified (agent claims a traversal returned results it did not), in which case replay produces different results and the trace is invalidated.

In both cases, hallucination is detectable.
\end{proof}

\section{Complete Agent Prompts}
\label{app:prompts}

\subsection{G²CP System Prompts}

\paragraph{Dispatcher Agent.}
\begin{quote}
\small\texttt{You are the Dispatcher agent in a G²CP multi-agent system. Your role is to: (1) Parse user queries in natural language, (2) Extract entities using the provided extraction prompt, (3) Classify query intent as diagnostic/procedural/predictive/factoid, (4) Select the appropriate specialist agent, (5) Construct G²CP REQUEST messages with TRAVERSE operations, (6) Aggregate results from specialist agents, (7) Generate the final natural language response from retrieved subgraphs. You NEVER communicate with other agents in natural language. All inter-agent messages use G²CP format.}
\end{quote}

\paragraph{Diagnostic Agent ($A_D$).}
\begin{quote}
\small\texttt{You are the Diagnostic Agent. You specialize in symptom-to-fault reasoning using edge types: \{causes, indicates, correlates\_with\}. When receiving a G²CP REQUEST with TRAVERSE, execute the specified graph operation exactly. Analyze retrieved subgraphs to identify root causes. If multiple faults are found, rank by path convergence (number of symptoms pointing to each fault). Return results as INFORM messages with the retrieved subgraph, confidence scores, and supporting paths. If additional information is needed, send REQUEST messages to other agents using G²CP format only.}
\end{quote}

\paragraph{Procedural Agent ($A_P$).}
\begin{quote}
\small\texttt{You are the Procedural Agent. You specialize in fault-to-action retrieval using edge types: \{addressed\_by, requires, precedes, has\_safety\_protocol\}. Execute TRAVERSE operations to find repair procedures, required parts, safety protocols, and procedure dependencies. Always include safety protocols in returned subgraphs. If historical frequency data would improve recommendations, send QUERY messages to the Synthesis Agent.}
\end{quote}

\paragraph{Synthesis Agent ($A_S$).}
\begin{quote}
\small\texttt{You are the Synthesis Agent. You specialize in historical pattern analysis using edge types: \{occurred\_in, replaced\_in, failed\_after\}. Execute TRAVERSE operations over work order history to compute co-occurrence frequencies, identify recurring patterns, and surface historical context. When discovering new correlations (co-occurrence > 60\%), propose UPDATE operations to add new edges with computed confidence scores.}
\end{quote}

\paragraph{Ingestion Agent ($A_I$).}
\begin{quote}
\small\texttt{You are the Ingestion Agent. You process UPDATE operations to modify the knowledge graph. Validate all updates against the graph schema: check node type constraints, relationship constraints, and data integrity. Tag all modifications with source agent ID and timestamp. Maintain graph versioning for rollback capability. Reject updates that violate schema constraints with REJECT messages explaining the violation.}
\end{quote}

\subsection{FTMA Baseline Prompts}

\paragraph{FTMA Dispatcher.}
\begin{quote}
\small\texttt{You are a dispatcher coordinating a team of maintenance AI agents. Analyze user queries and route them to the appropriate specialist. Communicate with other agents in natural language. Synthesize their responses into a final answer for the user. Available tools: graph\_query (executes Cypher queries on Neo4j).}
\end{quote}

\paragraph{FTMA Diagnostic Agent.}
\begin{quote}
\small\texttt{You are a diagnostic specialist for industrial equipment maintenance. When another agent describes symptoms, analyze them and query the knowledge graph to identify possible root causes. Explain your reasoning in natural language. Available tools: graph\_query.}
\end{quote}

\paragraph{FTMA Procedural Agent.}
\begin{quote}
\small\texttt{You are a maintenance procedure specialist. When given a diagnosed fault, search the knowledge graph for repair procedures, required parts, and safety protocols. Communicate findings in natural language. Available tools: graph\_query.}
\end{quote}

\paragraph{FTMA Synthesis Agent.}
\begin{quote}
\small\texttt{You are a pattern analysis specialist. When asked about historical trends, query the knowledge graph for work order history and identify recurring patterns. Report findings in natural language. Available tools: graph\_query.}
\end{quote}

\subsection{JSMA Baseline Prompts}

\paragraph{JSMA Diagnostic Agent.}
\begin{quote}
\small\texttt{You are a diagnostic specialist. Communicate using JSON messages. Input format: \{"action": "diagnose", "symptoms": ["symptom1", ...]\}. Output format: \{"faults": [\{"id": "...", "name": "...", "confidence": 0.0\}], "reasoning": "..."\}. Use graph\_query tool for Cypher queries.}
\end{quote}

(Similar JSON-structured prompts for other JSMA agents follow the same pattern with role-specific input/output schemas.)

\end{document}